\documentclass[10pt,journal]{IEEEtran}

\usepackage{xcolor,soul,framed} %,caption
\UseRawInputEncoding
\colorlet{shadecolor}{yellow}
\usepackage[pdftex]{graphicx}
\graphicspath{{../pdf/}{../jpeg/}}
\DeclareGraphicsExtensions{.pdf,.jpeg,.png}
\usepackage{mathrsfs}
\usepackage[cmex10]{amsmath}
\usepackage{array}
\usepackage{mdwmath}
\usepackage{mdwtab}
\usepackage{eqparbox}
\usepackage{caption}
\usepackage{graphicx}
\usepackage{float} 
\usepackage{subcaption}
\usepackage{url}

\newtheorem{theorem}{\bf Theorem}
\newtheorem{corollary}{\bf Corollary}
\begin{document}

\def\QEDclosed{\mbox{\rule[0pt]{1.3ex}{1.3ex}}}
\def\QEDopen{{\setlength{\fboxsep}{0pt}\setlength{\fboxrule}{0.2pt}\fbox{\rule[0pt]{0pt}{1.3ex}\rule[0pt]{1.3ex}{0pt}}}}
\def\QED{\QEDopen}
\def\proof{}
\def\endproof{\hspace*{\fill}~\QED\par\endtrivlist\unskip}

\bstctlcite{IEEEexample:BSTcontrol}
    \title{RIS-aided Wireless Communications: Can RIS Beat Metal Plate?}
  \author{Jiangfeng Hu,
      Haifan Yin,
      Li Tan,
      Lin Cao,
      Xilong Pei

     %<-this % stops a space
\thanks{J. Hu, H. Yin, L. Tan, L. Cao  and X. Pei are with Huazhong University of Science and Technology, 430074 Wuhan, China. (e-mail: \{jiangfenghu, yin, ltan, caolin, pei\}@hust.edu.cn. {The corresponding author is Li~Tan.}}
\thanks{This work was supported in part by the National Key Research and Development Program of China under Grant 2020YFB1806904, in part by the National Natural Science Foundation of China under Grants 62071191, 62071192, and 1214110.}
}

% ====================================================================
\maketitle

% === ABSTRACT ====================================================================
% =================================================================================
\begin{abstract}
%\boldmath
Reconfigurable Intelligent Surface (RIS) has recently been regarded as a paradigm-shifting technology beyond 5G, for its flexibility on smartly adjusting the response to the impinging electromagnetic (EM) waves. Usually, RIS can be implemented by properly reconfiguring the adjustable parameters of each RIS unit to align the signal phase on the receiver side. And it is believed that the phase alignment can be also mechanically achieved by a metal plate with the same physical size. However, we found in the prototype experiments that, a well-rotated metal plate can only approximately perform as well as RIS under limited conditions, although its scattering efficiency is relatively higher. When it comes to the case of spherical wave impinging, RIS outperforms the metal plate even beyond the receiving near-field regions. We analyze this phenomenon with wave optics theory and propose explicit scattering models for both the metal plate and RIS in general scenarios. 
Finally, the models are validated by simulations and field measurements. 

\end{abstract}

% === KEYWORDS =====================================================
% =================================================================================
\begin{IEEEkeywords}
Reconfigurable intelligent surfaces, Scattering model, Wave optics theory.
\end{IEEEkeywords}

% === I. INTRODUCTION =============================================================
% =================================================================================
\section{Introduction}

\IEEEPARstart{R}econfigurable Intelligent Surface (RIS) is widely believed to be a promising technology for its capability of smartly configure the impinging EM waves with a large number of sub-wavelength-sized passive scattering units \cite{RISpop}. A series of application scenarios have been investigated under different assumptions \cite{liu2021survey}. However, the RIS technology is still in its infancy, since the actual physical characteristics, the efficient reconfiguration methodologies and  deployment strategies in practical scenarios have not been well-studied so far. Among these challenges, the physical characteristics remain urgent to be modeled and clarified, which are the basis of other effective extensions.

In the early literature, RIS is usually described as an anomalous mirror that reflects impinging plane waves as outgoing plane waves with an unnatural angle (which is different from the specular angle predicted by Snell's Law). Recently, it has been clarified by \cite{RISMythEmil} that neither the plane wave nor infinite surface area will typically appear in practice, thus they are regarded as theoretical idealizations for RIS. And the authors illustrated by simulations that an RIS of finite size can generally not be interpreted as a mirror when impinged by plane waves. 

It is interesting to think, what will actually happen if a spherical wave (or even a beam) is impinging on a finite-sized RIS? The authors of \cite{tang2021mmwave} validated by experimental measurements that an RIS configured to uniform state with high scattering efficiency has similar characteristics as a metal plate with the same shape and size. However, the transceivers in \cite{tang2021mmwave} are placed in the far-field of the RIS operating in millimeter-wave bands, and there is lack of the comparisons between the well-rotated metal plate and a smartly-configured RIS, which is a relatively more general application scenario for RIS \cite{liu2021survey}. In fact, some researchers suggest we may carefully employ ``fully-passive RIS" (e.g., the metal plate) in some specific scenarios where the EM environment is quasi-static to further reduce the hardware and energy overhead. The paper \cite{PathlossEmil} 
indicated that a well-rotated metal plate is as efficient as RIS under the far-field assumption. What if the array grows large enough relative to the communication distance thus we are in the near-field of RIS array? To the best of our knowledge, there is still a lack of experimental validations in the literature whether or to what extent a well-configured RIS can be replaced with a well-rotated ``fully-passive RIS" of the same shape and size, especially when impinged by spherical waves. 
In view of the conditions above, the main contributions of this paper are:
\begin{itemize}
    \item Based on the wave optics theory, we establish explicit scattering efficiency models for both the RIS and the metal plate, which are represented by the concept of radar cross section (RCS), and extend the standard bistatic RCS \cite{balanis2012advanced} of a rectangle plate to more general cases. 
    \item We verify the the proposed models by simulation and field measurements. The measurement results are in good agreement with the simulation results predicted by the proposed models.
    \item We validate by both simulation and measurement results that a well-configured RIS can outperform a well-rotated metal plate of the same size in most practical cases, which indicates the power-focusing superiority of RIS.
\end{itemize}
%in the location where is most suitable for the metal plate to assist signal enhancement
\emph{Notations:} We use uppercase and lowercase bold-face variables to denote matrices and vectors, respectively. ${{\bf{r}}} = {x}\hat{\bf{x}}+{y}\hat{\bf{y}}+{z}\hat{\bf{z}}$ is a vector with Cartesian coordinates $(x,y,z)$, and $ {\bf{\hat r}} = uni(\bf {r})$, which is the unitization form of $\bf {r}$ representing its direction. $\left\| {\bf{r}} \right\|$ and ${{\bf{r}}^*}$ denote the Frobenius norm and  the conjugate transpose of $\bf {r}$, respectively. ${\nabla}$ is the nabla operator, and $\nabla_{\mathbf{r}}$ denotes the gradient evaluated at the $\bf {r}$ vector. $\cdot$ and ${\times}$ denote the scalar and vector products between vectors, respectively, and ${\bf{I}}$ is the identity matrix.

% =================================================================================
\begin{figure}
\centering 
  \includegraphics[width=3.3in]{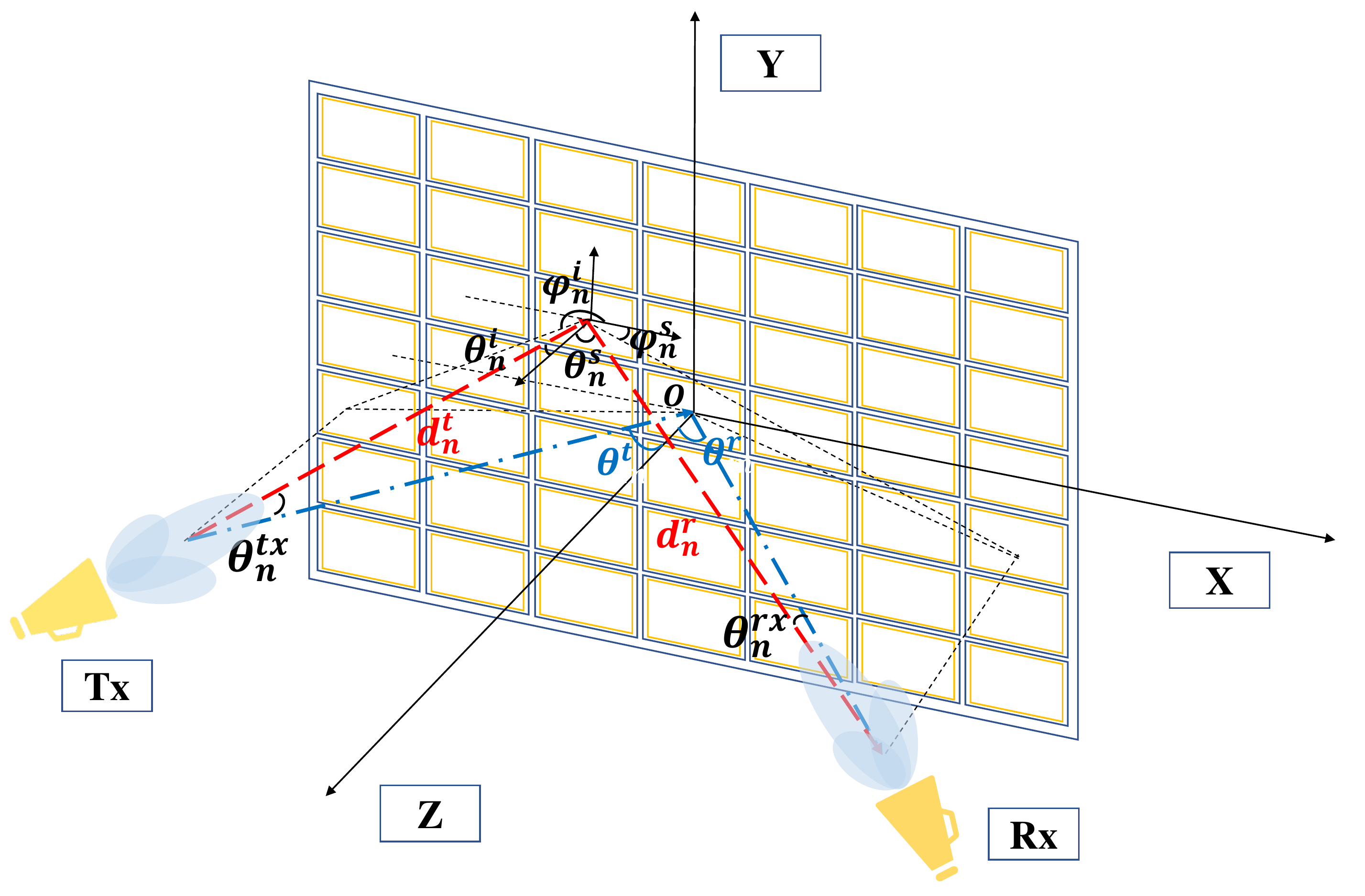}
 \caption{The geometric illustration of an RIS-aided communication system.} \label{Experimental skectchmap}
\end{figure}
\section{System model}
We consider an RIS-aided communication system in three-dimensional (3D) space. For ease of exposition and experiments, the transmitter (Tx) and the receiver (Rx) are both equipped with a single antenna. The generalization of the proposed model to multiple antenna settings is straightforward. The above settings are consistent in both simulation and field measurement experiments. The Tx, Rx and RIS are placed in a Cartesian coordinate system, as shown in Fig. (\ref{Experimental skectchmap}), whose locations  Tx and Rx are denoted as ${{\bf{p}}_t}$, ${{\bf{p}}_r}$ and ${{\bf{p}}_r}$, respectively. 
\par Without loss of generality, we place the RIS of a rectangle shape with ${N_v} \times {N_h}$ elements on the $xOy$-plane, with its geometric center aligned with the origin, where ${N_v}$ and ${N_h}$ denote the number of columns and rows of RIS elements. The two sides of the RIS are placed in parallel to the x-axis and y-axis. The element cells are designed to be edge-to-edge thus the vertical (horizontal) distance between their geometric centers (i.e., the spatial period) is exactly the width (length) of a single RIS element cell, denoted by ${d_v}$ and ${d_h}$, respectively. For more design details, we refer the interested readers to our previous work \cite{pei2021RIS5.8GHz}. We denote the $n$-th RIS element located at ${\bf{p}_{n}}$ by ${C_{n}}$, $n = 1,...,{N_v} {N_h}$. 

We assume the Tx and Rx are always in the element far-field, instead of the array far-field \cite{juan2020}. This usually makes sense at a normal communication distance, which is long enough relative to the sub-half-wavelength element size. For simplicity, the mutual coupling effect is not considered. Therefore, the impinging EM waves can be regarded (approximately) as plane waves for a particular RIS element and we can exploit the signals scattered by each independent element to calculate the signal at the receiver side. For a single RIS element ${C_{n}}$, we model its scattering factor ${S_{n}}$ as the combination of an inherent part and a reconfigurable part as:
\begin{equation}\label{combination scattering factor}
{S_{n}\left( {{\bf{\hat r}}_{n}^i,{\bf{\hat r}}_{n}^s}, {u_{n}}\right)} = f_{n}\left( {{\bf{\hat r}}_{n}^i,{\bf{\hat r}}_{n}^s} \right) {R_{n}}({u_{n}}),
\end{equation} where $f\left( {{\bf{\hat r}}_{n}^i,{\bf{\hat r}}_{n}^s} \right)$ is the inherent part (i.e., the original element scattering characteristics) formulated by a bidirectional scattering distribution (BSD) function which will be derived and discussed in details later. ${\bf{\hat r}}_{n}^i = uni({{\bf{p}}_{n}} - {{\bf{p}}_t})$, ${\bf{\hat r}}_{n}^s = uni({{\bf{p}}_r} - {{\bf{p}}_{n}})$ denote the direction of the incident wave, the direction of the scattered wave, respectively. The reconfigurable part
\begin{equation}\label{reconfig part}
{R_{n}}({u_{n}}) = {\alpha _{n}}({u_{n}}){e^{ - j{\phi _{n}}({u_{n}})}}
\end{equation}
is determined by the on-board control signal ${u_{n}}$ (e.g., a biasing voltage that can be properly selected), where ${\alpha _{n}}$ and ${\phi _{n}}$ represent the amplitude and phase response, respectively.
Therefore, if we denote the transmit signal by ${x}$, the received signal reflected by a single RIS element ${C_{n}}$ can be formulated as $y_{n} = {g_{n}} {S_{n}} {h_{n}} {x}$, where  ${h_{n}}$ and  ${g_{n}}$ denote the channel coefficient between the Tx and ${C_{n}}$, the channel coefficient between ${C_{n}}$ and the Rx, respectively. Take ${h_{n}}$ for example, the channel coefficient can be formulated by
${h_{n}} = {\beta _{n}}{e^{ - j2\pi \frac{{d{t_{n}}}}{\lambda }}}$, where $\lambda$ is the wave length and ${d_{n} ^t}$ is the distance between the Tx and the RIS element ${C_{n}}$. And ${\beta _{n}} = \sqrt {{{{\beta _0}\cos \theta _{n}^{tx}}}/{4 \pi({d_{n} ^t}) ^\gamma}}$ represents the path-loss factor,  where ${{\beta _0}}$ and $\gamma$ are constant depending on the communication environment, $\theta _{n}^{tx}$ is the angle between Tx-O and Tx-${C_{n}}$ to account for the effect of antenna directivity \cite{wankai2020model}, as shown in Fig. \ref{Experimental skectchmap}. The channel coefficient ${g}_{n}$ can be formulated similarly.

\section{MODELING OF THE RECEIVED SIGNAL}

In this section, we model and analyze the received signal that is reflected by either a metal plate or an RIS. In the following context, we adopt the term {\it{diffraction}} to emphasize the scattering at the edges of an element cell. And these edges usually occur when material discontinuity exists, which may lead to the spatial distribution discontinuities of permittivity and permeability.  

As analyzed in \cite{knott2004rcsbook}, for a perfectly conducting rectangular plate, the radar cross section is an appropriate parameter to depict its scattering efficiency in the given directions. RCS is originally defined as the ratio of the incident and scattered electric field power (or equivalently the magnetic field power) and can be formulated as: ${\sigma} = \mathop {\lim }\limits_{R \to \infty } \left( {4\pi {R^2}{{{{\left| {{{\bf{E}}^s}} \right|}^2}}}/{{{{\left| {{{\bf{E}}^i}} \right|}^2}}}} \right)$, where ${\bf{E}}^i$ and ${\bf{E}}^s$ denote the incident and scattering electric field, respectively, and $R$ is the distance between the source and the geometric center of the plate. 
\par The formula above inherits the far-field assumption. Nevertheless, one of the main advantages of RIS is the nearly-passive characteristic, which enables it to have a larger array size than the massive MIMO setups to interact with more EM waves propagating in wireless environment \cite{ScalingLawsEmil}. Therefore, sometimes the transceivers will inevitably be in the near field of the whole RIS array, and a direct application of the definition above may not be always appropriate in our experiments. Fortunately, though we may not always in the far field of the whole surface, we will be in the far field of the element in most cases. For fair comparison, we discretize the metal plate in the same way as our designed RIS prototype (i.e., same number of cell, equal cell size) and calculate the RCS on each element cell separately before summing them up. Note that it is a common operation in the literature of RCS, where a continuous surface is divided into a collection of small discrete patches \cite{radarhandbook}. Under these assumptions, first we have the following theorem for the metal cell: 
\begin{theorem}\label{theorem1}
\emph{Within the element far-field regions, the RCS of a single metal cell ${M_{n}}$ under the given incident angle and observation angle is formulated as:}
% \vspace{-0.3cm}
\begin{equation}\label{metal lemma}
\begin{aligned}
&\sigma_{n}^{M}\left(\hat{\mathbf{r}}_{n}^{i}, \hat{\mathbf{r}}_{n}^{s}\right)=4 \pi\left({d_{v} d_{h}}/{\lambda}\right)^{2}  \cos ^{2} \theta_{n}^{i}  \cdot  \\
&\left(\cos ^{2} \theta_{n}^{s} \cos ^{2} \varphi_{n}^{s}+\sin ^{2} \varphi_{n}^{s}\right)\left({\sin {(X)}}/{X}\right)^{2}\left({\sin {(Y)}}/{Y}\right)^{2}.
\end{aligned}
\end{equation}
\end{theorem}
where $\theta _{n}^i$ $({\varphi _{n}^i})$, $\theta _{n}^s$ $({\varphi _{n}^s})$ denote the elevation (azimuth) angle of the incident wave, the elevation (azimuth) angle of the scattered wave relative to the element ${C_{n}}$, respectively. $X$ and $Y$ are defined in (\ref{metal X}) and (\ref{metal Y}), respectively.

\begin{proof}
\quad \emph{Proof:} The derivations are given in Appendix A. 
\end{proof}
\par It can be inferred from the result above that the RCS of a metal cell is related to the wave propagation direction, the observation direction, and the metal cell inherent properties like its aspect ratio and the size relative to the wavelength. In addition, once the direction of observation is determined, the scattered power is relatively larger near the specular direction  (i.e., $\theta _{n}^i= \theta _{n}^s $, $\varphi _{n}^i = \varphi _{n}^s + \pi $) and the scattering efficiency will increase when the direction is closer to the boresight of the plate. Thus the received power will be maximized when these two directions are matched, i.e., the waves are transmitted and received along the boresight of the plate.
\par Since RCS is a parameter originally describing the attenuation ratio between the incident and scattered power, its effect on the transmitting signal can be derived as $f_{n}^M\left( {{\bf{\hat r}}_{n}^i,{\bf{\hat r}}_{n}^s} \right) = \sqrt {\sigma _{n}^M\left( {{\bf{\hat r}}_{n}^i,{\bf{\hat r}}_{n}^s} \right)}$, which is the BSD function we mentioned in the previous section. 
Thus, the receive signal scattered by the whole metal plate can be formulated as 
\begin{equation}\label{total received signal}
{y^M} =  {\sum\limits_{n = 1}^{{N_v} {N_h}} {{h_{n}}f_{n}^M\left( {{\bf{\hat r}}_{n}^i,{\bf{\hat r}}_{n}^s} \right){g_{n}}} }{x}.
\end{equation}

When it comes to the case of RIS, the element structure is relatively more complicated than that of a metal cell, which leads to a more complex formulation representing its scattering characteristics. Specifically, on the surface of an RIS with practically discretized configurations, there are inevitably material discontinuities between copper coats and substrates within elements \cite{pei2021RIS5.8GHz}, which will cause spatial distribution discontinuities of the permittivity and permeability \cite{Comments1996RCS}, resulting in the edge diffraction of the incident waves. What is more, the control signal (e.g., external bias voltage for varactor diode based-RIS) applied on each RIS element will bring an extra impedance change. Thus, compared with a metal cell, additional phase and amplitude responses will appear on the RIS element. However, according to our design experience, the RIS element structure is usually designed flexibly to adapt to the working frequency bands and application requirements. 
In Corollary \ref{Corollary1}, we provide an adjusted RCS for RIS element in order to establish a simple yet general scattering model:
\begin{corollary}\label{Corollary1}
\emph{Within the element far-field regions, the RCS of a single RIS element $C_{n}$ under the given incident angle and observation angle can be formulated as:}
{\begin{equation} \label{RIS lemma}
\begin{aligned}
&\sigma_{n}^{R}\left(\mathbf{r}_{n}^{i}, \mathbf{r}_{n}^{s}\right)=4 \pi\left({{d_{v} d_{h}}}/{\lambda}\right)^{2}\left({\sin {(X)}}/{X}\right)^{2}\left({\sin {(Y)}}/{Y}\right)^{2} \cdot\\
& {(\cos ^{2} \theta_{n}^{i})} \left(\cos ^{2} \theta_{n}^{s} \cos ^{2} \varphi_{n}^{s}+\sin ^{2} \varphi_{n}^{s}\right) \mathscr{D}_{n}\left(\mathbf{r}_{n}^{i}, \mathbf{r}_{n}^{s}\right)
\end{aligned}
\end{equation}} where
\begin{equation} \label{RIS diffraction factor}
\begin{aligned}
&\mathscr{D}_{n}\left(\mathbf{r}_{n}^{i}, \mathbf{r}_{n}^{s}\right)= 1 \ - \\
&\mu \sin \left(({\theta_{n}^{i}+\theta_{n}^{s})}/{2}\right) \cos \left({k} d _v \left({(\sin \theta_{n}^{i}+\sin \theta_{n}^{s})}/{2}\right)\right).
\end{aligned}
\end{equation} $k$ is the wavenumber, ${\mu}$ is the diffraction loss factor describing energy loss ratio due to diffraction (depending on the medium types at the edges), and the meanings of other parameters are aligned with Theorem \ref{theorem1}. 
\end{corollary}
\begin{proof}
\quad \emph{Proof:} Based on eq. (\ref{metal lemma}), it is proposed  with an additional diffraction component to depict the edge diffraction effects of the RIS element, which is inspired from the first principle in \cite{ross1966angle}: the diffraction effects mainly depend on the elevation angles, which are more evident when approaching grazing direction (i.e., $\theta _{n}^i = \pi /2$ and $\theta _{n}^s = \pi /2$).
\end{proof}
\begin{figure}[t]
\centering 
  \includegraphics[width=3in]{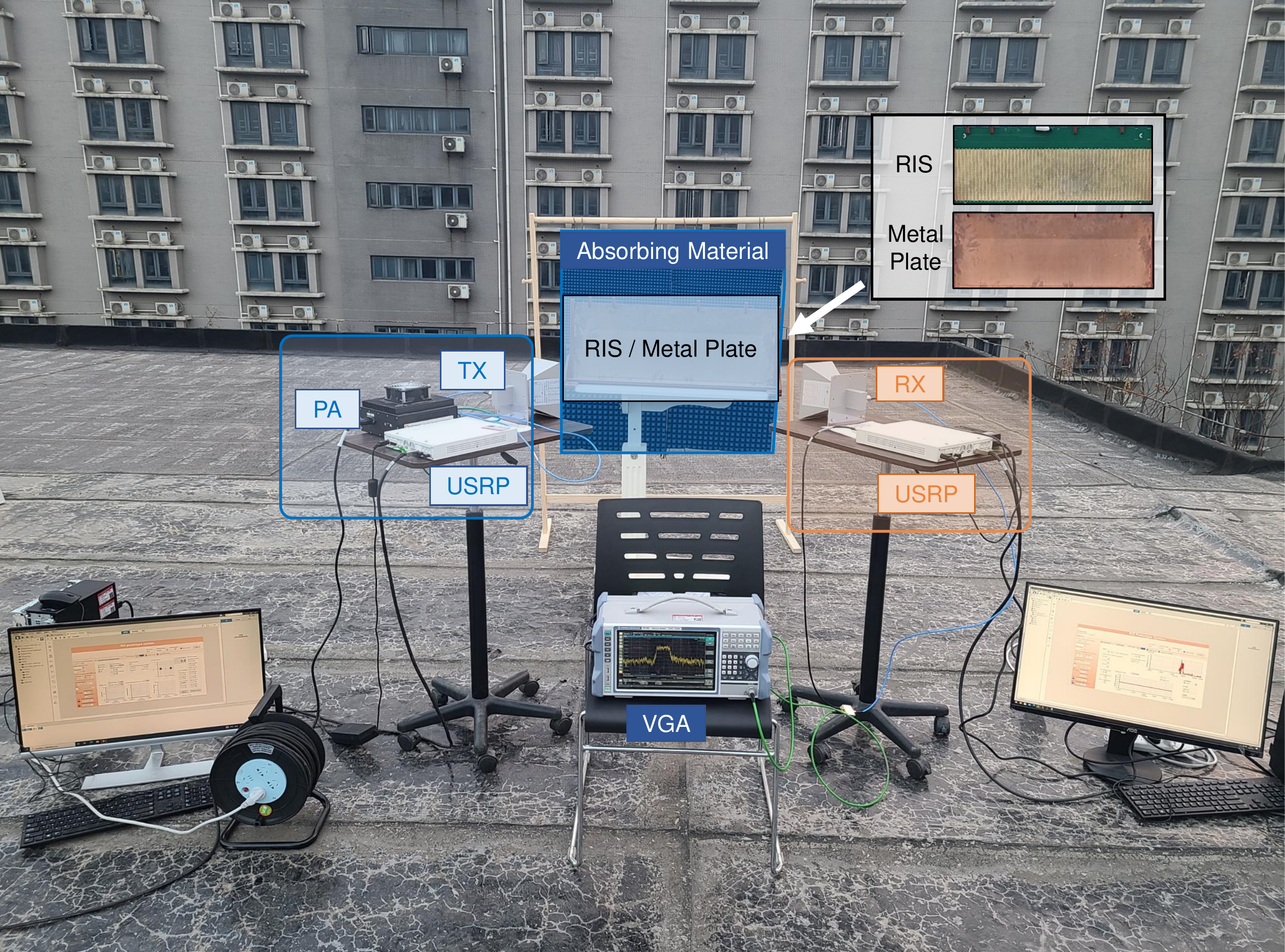}
 \caption{The testing and measurement scenarios.}\label{Experimental site}
\end{figure}
Accordingly, the BSD of the RIS element $C_{n}$ is derived as $f_{n}^R\left( {{\bf{r}}_{n}^i,{\bf{r}}_{n}^s} \right) = \sqrt {\sigma _{n}^R\left( {{\bf{r}}_{n}^i,{\bf{r}}_{n}^s} \right)}$. Finally, the received signal scattered by the whole RIS plane is:
\begin{equation}\label{RIS received signal}
\begin{array}{l}
{y^R} = {\sum\limits_{n = 1}^{{N_v} {N_h}} {{h_{n}}{R_{n}}({u_{n}})} } {\mathop{f}_{n}^R}\left( {{\bf{r}}_{n}^i,{\bf{r}}_{n}^s} \right){g_{n}}{x}.
\end{array}
\end{equation} 
% =================================================================================
\begin{figure*}[t]
    \begin{subfigure}[t]{0.32\textwidth}
           \centering
           \includegraphics[width=\textwidth]{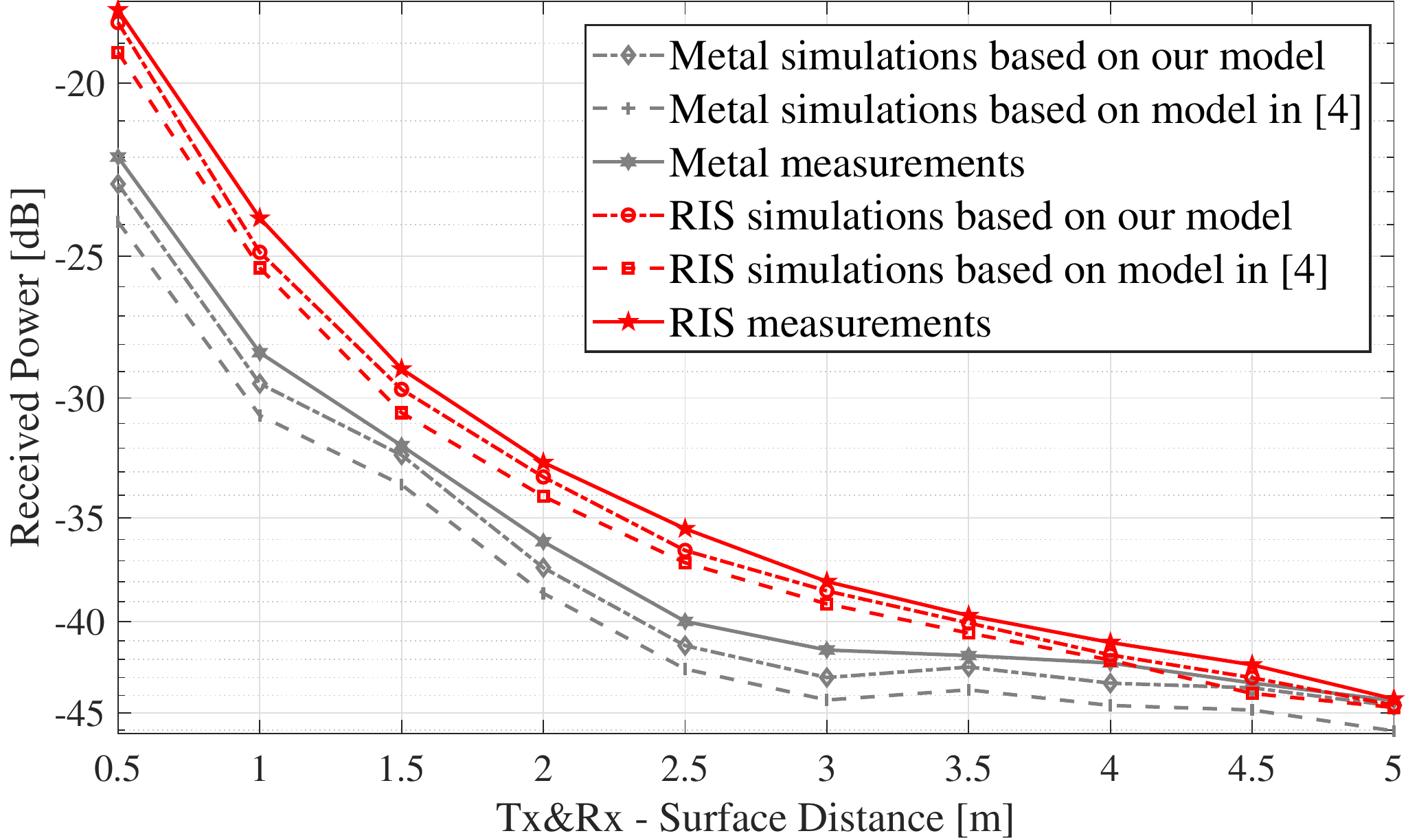}
            \caption{}
            \label{distance_30du}
    \end{subfigure}
    \begin{subfigure}[t]{0.32\textwidth}
            \centering
            \includegraphics[width=\textwidth]{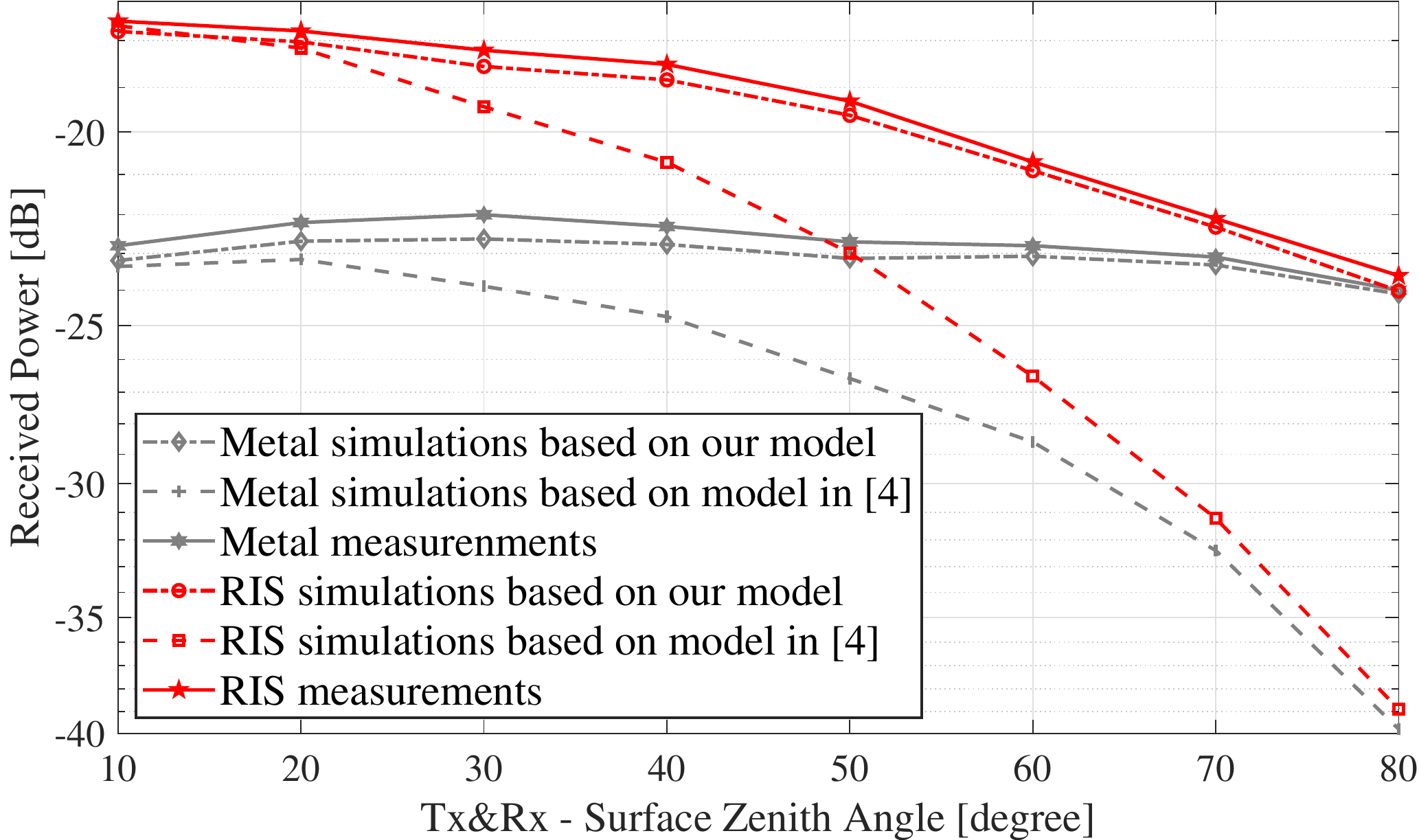}
            \caption{}
            \label{sweep_angle_50cm}
    \end{subfigure}
    \begin{subfigure}[t]{0.32\textwidth}
            \centering
            \includegraphics[width=\textwidth]{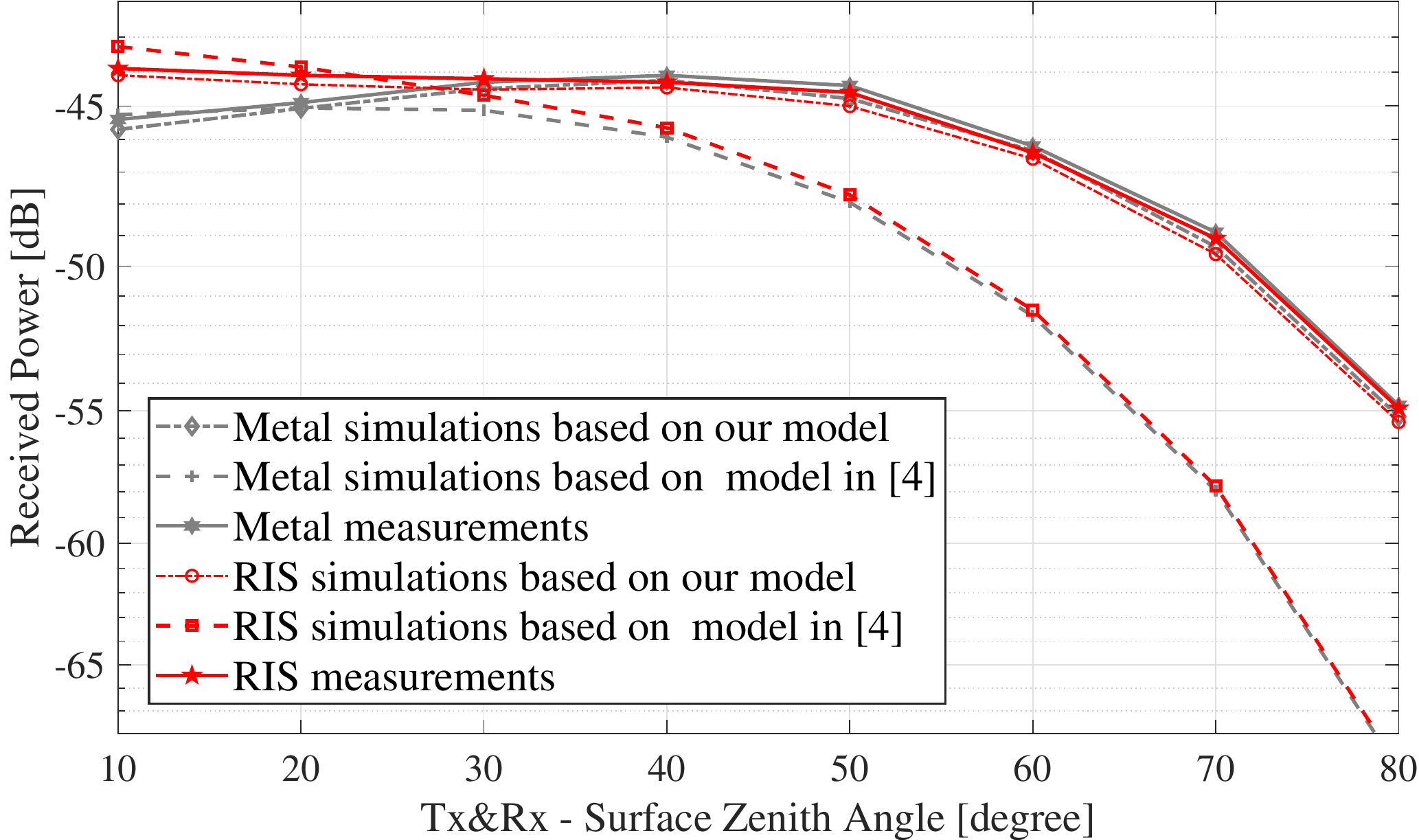}
            \caption{}
  \label{sweep_angle_5m}
    \end{subfigure}
    \caption{(a) Received power versus the Tx\&Rx - Surface distances at 30 degree zenith angle. (b) Received power versus the zenith angles, the Tx\&Rx - surface distances are 50 cm. (c) Received power versus the zenith angles, the Tx\&Rx-surface distances is 5 m. }
\end{figure*}
\section{Experimental Results and Analysis}
In this section, we start by briefly introducing our experimental configuration and then validate the power focusing capability of RIS, as well as our proposed models of the received signal. 
Fig. \ref{Experimental site} illustrates the measurement environments. The basic parameters of our field measurements and details of the RIS platform can be found in our previous work \cite{pei2021RIS5.8GHz}. Before the experiments, we have calibrated the system by aligning the transmit and receive antenna under a line-of-sight setting.

Since the configuration goal is to focus signals at the receiver side, here we iterate the algorithm introduced in \cite{pei2021RIS5.8GHz} for several times to approach the optimal configuration.
% Note that the quantization of ${R_{n}}({u_{n}})$ will actually lead to some performance loss compared to high-precision configuration, yet we found both in simulations and the field experiments that the 1-bit quantized RIS can still perform better than the metal plate to focus signals in many cases.
%${d_{tx}}$ and ${d_{rx}}$ distance between Tx and the origin, the distance between Rx and the origin,
To prove experimentally the power focusing superiority of RIS, we compare the received power with the assistance of RIS and that of a same-sized well-rotated metal plate under several different scenarios. As depicted in Fig. \ref{Experimental skectchmap}, ${\theta ^{t}}$ and ${\theta ^{r}}$ denotes the zenith angle from the origin to the Tx, the zenith angle from the origin to the Rx, respectively. Since we are interested in the comparisons that are most effective for the metal plate to enhance the receiving signals, we always set the two angles to be identical in the subsequent content, as shown in Fig. \ref{Experimental site}.
%horizontal mid-perpendicular plane of the surfaces

The model proposed in \cite{tang2021mmwave} considered several factors that impact the received power in RIS-assisted system, which is suitable to be adopted here for comparisons. Their normalized power radiation pattern for each unit can be equivalently rewritten in the form of RCS as \cite{tang2021mmwave} (notations have been modified to be consistent with our paper) : $\sigma _{n}^C \left(\hat{\mathbf{r}}_{n}^{i}, \hat{\mathbf{r}}_{n}^{s}\right)= {\cos ^2}\theta _{n}^i{\cos ^2}\theta _{n}^s.$ Thus, by squaring the received signal and adding the normalization factor, the total received power can be formulated as in \cite{Marco2022model}:
\begin{equation}
{P_r} = \frac{{{P_t}{\lambda ^2}}}{{4\pi }}{\left|{ {\sum\limits_{n = 1}^{{N_v} {N_h}} {{h_{n}}{R_{n}}} } \sqrt {{\sigma _{n}}} {g_{n}}} \right|^2}
\end{equation}
\par For fair comparison, other factors except for the RCS item are set to be the same for different models in [4] and ours. And for the metal plate, the coefficient ${R_{n}}$ is set to 1.
\par Fig. 3(\subref{distance_30du}) shows the measurement results and the simulation results based on the proposed models at the specular zenith angle of 30 degree (i.e., ${\theta ^{t}}={\theta ^{r}}={30^ \circ }$) and the same Tx\&Rx-surface distances. It mainly illustrates the relationship between the received power and the communication distances. We can observe that within the border between near field and far field of RIS ($\frac{{2{N_v}{N_h}{d_v}{d_h}}}{\lambda }  \simeq 6$ m) \cite{wankai2020model}, as the Tx-surface and Rx-surface distances increase simultaneously, the received power of the RIS-aided system is always higher due to the focalization of RIS, yet decreases faster than the metal-aided one. The focusing superiority of RIS gradually disappears as the distances approach 5 m, which makes sense since the phase shifts (seen from both the Tx and the Rx) between the surface center and edges will decrease rapidly as the distances increase. Moreover, the relative side lobe level (RSLL) of a well-configured RIS will decrease with distance, which also impacts the focusing ability of it. The measurement results have similar trend with simulation results based on our model, however with a slight increase which might result from the multi-path effect caused by scatterers in the environment. 
\par Fig. 3(\subref{sweep_angle_50cm}) shows the measurement results and the simulation results based on the proposed models at Tx\&Rx - surface distances of 50 cm. It mainly illustrates the relationship between the zenith angles and the received power. We can observe that the received power of both RIS and metal plate assisted system decrease as the zenith angles deviated from the boresight and the former decreases faster. It is in line with our expectation since their scattering efficiency (represented by the RCS) will decrease in a similar manner. Within such a short distance, the focusing ability of RIS is obvious for most angles. This is because the phase differences over the scattering surface are relatively large, and most of the transmitting power can be utilized here, thus the reconfiguable phase shifts (which determine the actual scattering direction) of RIS is more important than a little higher RCS of the metal plate.

When we increase the Tx-surface and RX-surface distances to 5 m, as shown in Fig. 3(\subref{sweep_angle_5m}), it is interesting that the focusing ability of RIS still exists at small zenith angles. And the received power curves of the RIS and the metal plate aided systems begin to cross as the angle increases. Several causes may lead to this result. First, the maximum scattering efficiency (usually occurs in a small zenith angle) of RIS and the metal plate is different. Second, as mentioned above, the scattering efficiency of RIS will decrease faster as the zenith angles increase due to its edge diffraction phenomenon. Furthermore, at the relatively longer distances, the phase differences are more difficult to distinguish, which is not beneficial for RIS to focus signals.

% ===================================================================================================================================
% ===================================================================================================================================
\section{Conclusion}
In this paper, novel scattering models based on wave optics theory have been derived and explicitly formulated for both the metal plate and the RIS. The simulation results from these models showed good agreements with the field measurement results, which indicated the effectiveness of the proposed models. And we validated the power focusing superiority of a well-configured RIS by both simulation and measurement results, proved that it outperforms a well-rotated metal plate with the same shape and size in most cases, especially at relatively close distance and small zenith angle from the transceivers.

%\section*{Acknowledgment}

%Dr. Reveryrand would like to acknowledge the funding by XLIM, Limoges, France. 
%The authors would like to thank 

% if have a single appendix:
%\appendix[Proof of the Zonklar Equations]
% or
%\appendix  % for no appendix heading
% do not use \section anymore after \appendix, only \section*
% is possibly needed

% use appendices with more than one appendix
% then use \section to start each appendix
% you must declare a \section before using any
% \subsection or using \label (\appendices by itself
% starts a section numbered zero.)
%

\appendices
\section{Proof of Theorem \ref{theorem1}}
\begin{proof}
\emph{Proof:}
In the following, footmarks indicating specificis units are temporarily dropped to make the text more concise. Consider a point source located at ${{\bf{p}}_t} = {x_t}\hat{\bf{x}}+{y_t}\hat{\bf{y}}+{z_t}\hat{\bf{z}}$ generates an incident electric field $\mathbf{E}^{i n}\left(\mathbf{p}_{t}, \mathbf{r}\right)$ at the $n$-th element cell located at ${{\bf{r}}} = {x}\hat{\bf{x}}+{y}\hat{\bf{y}}+{z}\hat{\bf{z}}$.  The incident wave can be given by the tensor Green's function $\mathbf{G}(\mathbf{r})$, the resulting wave at $\bf{r}$ is
\begin{equation}
\mathbf{E}^{i n}\left(\mathbf{p}_{t}, \mathbf{r}\right)=\mathbf{G}\left(\mathbf{r}-\mathbf{p}_{t}\right) \mathbf{J}^{t}\left(\mathbf{p}_{t}\right),
\end{equation} where $\mathbf{J}^{t}\left(\mathbf{p}_{t}\right)$ is the monochromatic source current. And since the EM waves our practical system adopt is linearly polarized wave, we accordingly assume here only the $x$ direction of the current is excited at the source.
% transverse magnetic (TM) waves
Due to the fact that we may always in the far field of a sub-wavelength element cell, the wave arrives at the element is approximately a plane wave, thus we mainly focus on its phase shift over the element and assume a uniform amplitude $E_{0}^{in}$. Let $\mathbf{k}\left(\mathbf{p}_{t}, \mathbf{r}\right)$ denote the wavevector of the incident wave:
\begin{equation}
\mathbf{k}\left(\mathbf{p}_{t}, \mathbf{r}\right)=k\left(\sin \theta_{i} \cos \varphi_{i} \hat{\mathbf{x}}+\sin \theta_{i} \sin \varphi_{i} \hat{\mathbf{y}}-\cos \theta_{i} \hat{\mathbf{z}}\right),
\end{equation} which identifies the direction of the incident wave and $\theta_{i}$ ($\varphi_{i}$) denote the elevation (azimuth) angle of it. Therefore, the incident electric field can be formulated as
\begin{equation} \label{Eplane}
\mathbf{E}^{i n}\left(\mathbf{p}_{t}, \mathbf{r}\right)=E_{0}^{i n} \mathrm{e}^{-j k\left(\sin \theta_{i} \cos \varphi_{i} x+\sin \theta_{i} \sin \varphi_{i} y-\cos \theta_{i} z\right)} \hat{\mathbf{x}}.
\end{equation}
% $ \frac{1}{\eta_{0}}=\frac{k_{0}}{\omega \mu_{0}} $
Then, by applying Maxwell's equations, the incident magnetic field can be obtained as
\begin{small}
\begin{equation}
\begin{aligned}
\mathbf{H}^{i n}\left(\mathbf{p}_t, \mathbf{r}\right)&=-\frac{1}{j \omega \mu_0} \nabla_{\mathbf{r}} \times \mathbf{E}^{i n}\left(\mathbf{p}_t, \mathbf{r}\right) \\
& =\begin{small}{-\frac{E_0^{i n}}{\eta_0}\left(\cos \theta_i \hat{\mathbf{y}}+\sin \theta_i \sin \varphi_i \hat{\mathbf{z}}\right) }\end{small} \\ & \cdot \mathrm{e}^{-j k\left(\sin \theta_i \cos \varphi_i x+\sin \theta_i \sin \varphi_i y-\cos \theta_i z\right)},
\end{aligned}
\end{equation}
\end{small}where $\omega$ is the angular frequency of the wave,  $\mu_{0}$ and $\eta_{0}$ denote the permittivity and impedance of free-space, respectively.
Note that the element cells are all located on the $xOy$ plane (i.e., $z=0$), and we temporarily drop the arguments in $\mathbf{H}^{i n}\left(\mathbf{p}_{t}, \mathbf{r}\right)$ to lighten the notations:
\begin{small}
\begin{equation} 
\begin{aligned}
\mathbf{H}^{i n} &=\left.\mathbf{H}^{i n}\left(\mathbf{p}_{t}, \mathbf{r}\right)\right|_{z=0} \\
&=-\frac{E_{0}^{in}}{\eta_{0}}\left(\cos \theta_{i} \hat{\mathbf{y}}+\sin \theta_{i} \sin \varphi_{i} \hat{\mathbf{z}}\right) \mathrm{e}^{-j k\left(\sin \theta_{i} \cos \varphi_{i} x+\sin \theta_{i} \sin \varphi_{i} y\right)}.
\end{aligned} 
\end{equation} \end{small}Since the metal plate is a kind of good conductor whose thickness is negligible here, the equivalent current density at the incident point can be approximated as twice the amplitude of the incident tangential magnetic field components \cite{radarhandbook}:
\begin{small}
\begin{equation}
\begin{aligned}
\mathbf{J}^{s} & \simeq 2 \hat{\mathbf{n}} \times\left.\mathbf{H}^{in}\right|_{x=x^{\prime}  \atop y=y^{\prime}}  \\
&=2 \frac{E_{0}^{in}}{\eta_{0}} \cos \theta_{i} \mathrm{e}^{-j k\left(\sin \theta_{i} \cos \varphi_{i} x^{\prime}+\sin \theta_{i} \sin \varphi_{i} y^{\prime}\right)} \hat{\mathbf{x}},
\end{aligned}
\end{equation}
\end{small}where $\hat{\mathbf{n}}$ is the the unit surface normal. Therefore, by applying the above equation and adopting the auxiliary variables of vector potential defined in \cite{balanis2012advanced}, we can obtain that
\begin{small}
\begin{equation}
\begin{aligned}
N_{\theta}=&\left.\iint_{S}\left(J_{x}^{s} \cos \theta_{s} \cos \varphi_{s}+J_{y}^{s} \cos \theta_{s} \sin \varphi_{s}-J_{z}^{s} \sin \theta_{s}\right)\right|_{J_{y}^{s}=0 \atop J_{z}^{s}=0} \\
& \cdot e^{-j k\left(\sin \theta_{i} \cos \varphi_{i} x^{\prime}+\sin \theta_{i} \sin \varphi_{i} y^{\prime}\right)} d x^{\prime} d y^{\prime} \\
&= 2 {d _v} {d _h} \frac{E_{0}^{i n}}{\eta_{0}}\left[\cos \theta_{i} \cos \theta_{s} \cos \varphi_{s} ({\sin (X)}/{X}) ({\sin (Y)}/{Y})\right],
\end{aligned}
\end{equation}
\begin{equation}
\begin{aligned}
N_{\varphi}&=\left.\iint_{S}\left(-J_{x}^{s} \sin \varphi_{s}+J_{y}^{s} \cos \varphi_{s}\right)\right|_{J_{y}^{s}=0} \\
&\cdot e^{-j k\left(\sin \theta_{i} \cos \varphi_{i} x^{\prime}+\sin \theta_{i} \sin \varphi_{i} y^{\prime}\right)} d x^{\prime} d y^{\prime} \\
&= 2 {d _v} {d _h} \frac{E_{0}^{i n}}{\eta_{0}}\left[\cos \theta_{i} \sin \varphi_{s} ({\sin (X)}/{X})({\sin (Y)}/{Y})\right],
\end{aligned}
\end{equation} \end{small}where $\theta_{s}$ ($\varphi_{s}$) denote the elevation (azimuth) angle of the scattered wave, and
\begin{small}
\begin{equation} \label{metal X}
X=\frac{\pi {d _v}}{\lambda}\left(\sin \theta_{s} \cos \varphi_{s}+\sin \theta_{i} \cos \varphi_{i}\right),
\end{equation}
\begin{equation} \label{metal Y}
Y=\frac{\pi {d _h}}{\lambda}\left(\sin \theta_{s} \sin \varphi_{s}+\sin \theta_{i} \sin \varphi_{i}\right).
\end{equation}
\end{small}

Therefore, we can obtain the amplitude components of the scattered electric field as:
\begin{small}
\begin{equation}
\left|E_{\theta}^{s}\right|=\frac{k {d _v} {d _h} E_{0}^{i n}}{2 \pi r}\left[\cos \theta_{i} \cos \theta_{s} \cos \varphi_{s} ({\sin (X)}/{X}) ({\sin (Y)}/{Y})\right],
\end{equation}
\begin{equation}
\left|E_{\varphi}^{s}\right|=\frac{k {d _v} {d _h} E_{0}^{i n}}{2 \pi r}\left[\cos \theta_{i} \sin \varphi_{s} ({\sin (X)}/{X}) ({\sin (Y)}/{Y})\right].
\end{equation}
\end{small}
\par Finally, by adopting (\ref{Eplane}) and the original RCS definition, as well as \begin{small}$\left|E^{s}\right|=\sqrt{\left|E_{\theta}^{s}\right|^{2}+\left|E_{\varphi}^{s}\right|^{2}}$ \end{small}, we can obtain the RCS of a metal plate of size ${d _v} \times {d _h}$ as:
\begin{small}
\begin{equation}
\begin{aligned}
\sigma^{M}&=4 \pi\left(\frac{{d _v} {d _h}}{\lambda}\right)^{2}  \cos ^{2} \theta _{i}   \left(\cos ^{2} \theta _{s} \cos ^{2} \varphi _{s}+\sin ^{2} \varphi _{s}\right)\cdot \\
& \left({\sin ({X})}/{X}\right)^{2}\left({\sin ({Y})}/{Y}\right)^{2}.
\end{aligned}
\end{equation}
\end{small}
Thus Theorem \ref{theorem1} is proved.
\end{proof}

% ====== REFERENCE SECTION

%\begin{thebibliography}{1}

% IEEEabrv,

\bibliographystyle{IEEEtran}
\bibliography{IEEEabrv,Bibliography}
%\end{thebibliography}
% biography section

%% insert where needed to balance the two columns on the last page with
%% biographies
%%\newpage

%\begin{IEEEbiographynophoto}{Jane Doe}
%Biography text here.
%\end{IEEEbiographynophoto}
% ==== SWITCH OFF the BIO for submission
% ==== SWITCH OFF the BIO for submission

% You can push biographies down or up by placing
% a \vfill before or after them. The appropriate
% use of \vfill depends on what kind of text is
% on the last page and whether or not the columns
% are being equalized.

% Can be used to pull up biographies so that the bottom of the last one
% is flush with the other column.
%\enlargethispage{-5in}

% that's all folks
\end{document}